\documentclass[conference]{IEEEtran}
\IEEEoverridecommandlockouts

\usepackage{cite}
\usepackage{amsmath,amssymb,amsfonts}
\usepackage{algorithmic}
\usepackage{graphicx}
\usepackage{latexsym}
\usepackage{xcolor}
\usepackage{comment}
\usepackage{cite}
\usepackage{url}
\usepackage[T3,T1]{fontenc}
\DeclareSymbolFont{tipa}{T3}{cmr}{m}{n}
\DeclareMathAccent{\invbreve}{\mathalpha}{tipa}{16}
\usepackage{newtxtext}
\usepackage[varg]{newtxmath}
\usepackage{moreverb}
\usepackage{upgreek}
\usepackage{mathrsfs}
\usepackage{mathtools}
\usepackage{bm}
\def\BibTeX{{\rm B\kern-.05em{\sc i\kern-.025em b}\kern-.08em
    T\kern-.1667em\lower.7ex\hbox{E}\kern-.125emX}}

\newtheorem{theorem}{Theorem}
\newtheorem{lemma}{Lemma}
\newtheorem{proposition}{Proposition}
\newtheorem{property}{Property}
\newtheorem{remark}{Remark}
\newtheorem{definition}{Definition}
\newtheorem{corollary}{Corollary}

\allowdisplaybreaks

\newcommand{\wt}{\widetilde}

\newcommand{\expectation}{\mathrm{E}}

\newcommand{\phin}{\phi^{(n)}}
\newcommand{\phinX}{\phi^{(n)}(\bm{X})}
\newcommand{\phinx}{\phi^{(n)}(\bm{x})}

\newcommand{\psin}{\psi^{(n)}}

\newcommand{\Phin}{\Phi^{(n)}}
\newcommand{\PhiKn}{\Phi_{\bm{K}}^{(n)}}
\newcommand{\PhiKnX}{\Phi_{\bm{K}}^{(n)}(\bm{X})}
\newcommand{\PhiXnK}{\Phi_{\bm{X}}^{(n)}(\bm{K})}
\newcommand{\Phiknx}{\Phi_{\bm{k}}^{(n)}(\bm{x})}
\newcommand{\PhiknX}{\Phi_{\bm{k}}^{(n)}(\bm{X})}

\newcommand{\wtPhin}{\widetilde{\Phi}^{(n)}}

\newcommand{\wtPhiknx}{\widetilde{\Phi}_{\bm{k}}^{(n)}(\bm{x})}

\newcommand{\LPhiknX}{L_{\Phi_{\bm{k}}^{(n)}}(\bm{X})}

\newcommand{\Psin}{\Psi^{(n)}}
\newcommand{\PsiKn}{\Psi_{\bm{K}}^{(n)}}

\newcommand{\wtPsin}{\widetilde{\Psi}^{(n)}}

\newcommand{\Dn}{\mathcal{D}^{(n)}}

\newcommand{\MI}{\Delta_{\mathrm{MI}}^{(n)}}

\newcommand{\hugel}{{\arraycolsep 0mm
                    \left\{\ba{l}{\,}\\{\,}\ea\right.\!\!}}

\newcommand{\huger}{{\arraycolsep 0mm
                    \left.\ba{l}{\,}\\{\,}\ea\!\!\right\}}}

\newcommand{\hugebl}{{\arraycolsep 0mm
                    \left[\ba{l}{\,}\\{\,}\ea\right.\!\!}}

\newcommand{\hugebr}{{\arraycolsep 0mm
                    \left.\ba{l}{\,}\\{\,}\ea\!\!\right]}}

\newcommand{\defeq}{:=}

\newcommand{\beq}{\begin{equation}}
\newcommand{\eeq}{\end{equation}}
\newcommand{\beqa}{\begin{eqnarray}}
\newcommand{\eeqa}{\end{eqnarray}}
\newcommand{\beqno}{\begin{eqnarray*}}
\newcommand{\eeqno}{\end{eqnarray*}}
\newcommand{\ba}{\begin{array}}
\newcommand{\ea}{\end{array}}

\newcommand{\MEq}[1]{\stackrel{
{\rm (#1)}}{=}}

\newcommand{\MLeq}[1]{\stackrel{
{\rm (#1)}}{\leq}}

\newcommand{\MGeq}[1]{\stackrel{
{\rm (#1)}}{\geq}}

\newcommand{\MG}[1]{\stackrel{
{\rm (#1)}}{>}}

\newcommand{\MRar}[1]{\stackrel{
{\rm (#1)}}{\Rightarrow}}

\begin{document}

\title{ 
\newcommand{\isitaTitle}{
Variable-Length Source Encryption under 
Mutual Information Security Criterion:
Universal Coding, Strong Converse Theorem
}
\newcommand{\arXibTitle}{
}{
A Framework of Variable-Length Source Encryption
using Mutual Information Security Criterion:
Universal Coding, Strong Converse Theorem
}

\newcommand{\ieiceTitle}{
Secure Transmission of Codewords over the 
Noiseless Channel for the Variable-Length Source 
Coding and the Strong Converse Theorem 
}

}

\author{
	\IEEEauthorblockN{Yasutada Oohama and Bagus Santoso}
	\IEEEauthorblockA{University of 
    Electro-Communications, Tokyo, Japan\\ 
	Email: \url{{oohama,santoso.bagus}@uec.ac.jp}}
}
%

\newcommand{\authorZ}
{%
	\IEEEauthorblockN{Yasutada Oohama }
	\IEEEauthorblockA{University of Electro-Communications,\\ 
    Tokyo, Japan\\ 
	Email: oohama@uec.ac.jp}
    \and
    \IEEEauthorblockN{Morikazu Hamasaki}
	\IEEEauthorblockA{
    Oki Software Co., Ltd.,\\
    Tokyo, Japan\\ 
	Email: hamasaki274@oki.com}
    \and
	\IEEEauthorblockN{Bagus Santoso}
	\IEEEauthorblockA{University of Electro-Communications,\\ 
    Tokyo, Japan\\ 
	Email: santoso.bagus@uec.ac.jp}
}

\maketitle

\begin{abstract}
In this paper we consider the variable-length 
lossless source coding for discrete memoryless sources. 
We proposes a new encryption framework for 
securely transmitting codewords over a noiseless channel. 
The proposed source encryption framework is 
based on the secure communication framework of 
the Shannon cipher system. In the proposed framework, 
we use the mutual information as a measure of 
information leakage to an adversary.
We establish the necessary and sufficient condition for 
secure communication 
under the condition that 
the information leakage is upper bounded by 
a constant $\delta\in (0,\infty)$, 
thereby providing a 
complete solution to the problem. 
We also show that the obtained necessary 
and sufficient condition does not depend 
on the constant 
$\delta \in (0,\infty)$, demonstrating 
that we have the strong converse coding theorem 
for the proposed framework 
of source encryption.
We further prove the existence
of encryption/decryption schemes, which are 
universal in the sense that they work effectively 
for any distributions of the plain text and those 
of the key used for the encryption.
\end{abstract}


\section{Introduction}

Lossless source coding is widely used as 
a technique for compressing information in order 
to improve the efficiency of information 
transmission. In information transmission, 
in addition to such efficiency, communication 
security, namely, preventing information leakage
over public communication channels, is also 
an important issue. Shannon~\cite{Shannon} 
formulated cryptosystems within an 
information-theoretic framework, which 
we call Shannon cipher system or 
briefly say SCS.  For the SCS, Shannon clarified
the concept and fundamental limits of 
perfect secrecy for shared-key encryption schemes.
Since then, we have had a variety 
of studies extending SCS.  
Some important extensions were provided by 
Yamamoto~\cite{Yamamoto1}, \cite{Yamamoto2}
and Hayashi and Yamamoto~\cite{Yamamoto3}. 


In this paper we consider the variable-length 
lossless source coding for discrete memoryless sources. 
We proposes a new encryption framework for securely transmitting codewords over a noiseless channel. 
The proposed source encryption framework is based 
on the secure communication framework 
of SCS.  

In \cite{DBLP:conf/isit/OohamaS22} and
\cite{DBLP:conf/isit/OohamaS24},
Oohama and Santoso 
investigated common key cryptosystem 
involving side channel attacks.
In 
\cite{DBLP:conf/itw/OohamaS21},
\cite{DBLP:conf/isita/OohamaS22}, and 
\cite{Oohama2025} 
they further 
proposed some frame works for the distributed source  encryption and investigated necessary and 
sufficient conditions for achieving secure 
communication in the proposed frameworks. 

In this study, based on SCS \cite{Shannon} and the 
several previous works 
\cite{DBLP:conf/isit/OohamaS22},
\cite{DBLP:conf/isit/OohamaS24},
\cite{DBLP:conf/itw/OohamaS21},
\cite{DBLP:conf/isita/OohamaS22}, and 
\cite{Oohama2025} 
by Oohama and Santoso,  
we propose a framework of ``source encryption'' 
that applies encryption to a given source 
encoder and decoder. For the source encoder, 
we consider the case where a fixed-length sequence
generated from a source is transformed into 
a variable-length binary string, that is, 
the variable-length source coding. 
In the proposed framework, we use 
the mutual information as a measure of 
information leakage to an adversary.  
The maximum mutual information used in the 
previous works 
\cite{DBLP:conf/isit/OohamaS22},
\cite{DBLP:conf/isita/OohamaS22}
has an advantage that 
it provides a security measure determined solely 
by the encryption system, independent of the source 
characteristics.
This metric, however can not be applied 
to the proposed framework because 
in the proposed framework, the source encryption 
of variable-length codes must depend 
on the source statistics. 

For the proposed framework we establish 
the necessary and sufficient condition for 
secure communication 
under the condition that the 
information leakage is upper 
bounded by a constant $\delta\in (0,\infty)$, 
thereby providing a complete solution to the problem. 
We also show that the obtained necessary 
and sufficient condition does not depend 
on the constant $\delta \in (0,\infty)$, which implies  
that we have the strong converse 
coding theorem for the proposed framework 
of source encryption. Furthermore we prove 
the existence of encryption/decryption schemes, 
which are universal 
in the sense that they work effectively for any 
distributions of the plain text and those 
of the key used for the encryption.

As a parallel work with this study, the authors 
\cite{oohamaSa26FrwOfSEncUnivStConv
},  proposed a framework 
of source encryption built upon the fixed-length 
source coding and obtained the condition for reliable 
and secure communication in an explicit form. In the proofs 
of our main results we derive the sufficient 
condition by using the key results the authors obtained in \cite{oohamaSa26FrwOfSEncUnivStConv}.
For the proof of the necessary condition the 
derivation of this condition is quite 
standard, not depending on 
the work \cite{oohamaSa26FrwOfSEncUnivStConv}.

\section{
Secure Variable-Length Lossless 
Source Coding
}
\subsection{Preliminaries}
In this subsection, we show the basic notations 
and related consensus used in this paper.

\noindent 
\underline{\it Source of Information and Key:} \/ 
We first define the source. 
Let $X$ be a random variable from a finite 
set $\mathcal{X}$. Let $\{X_t\}_{t=1}^\infty$ be a 
stationary discrete memoryless 
source (DMS) such that for each $t=1,2,\ldots,$
$X_t$ takes values in the finite set $\mathcal{X}$ 
and has the same distribution as that of $X$ 
denoted by $p_X=\{p_X(x)\}_{x \in \mathcal{X}}$.
The stationary DMS 
$\{X_t\}_{t=1}^\infty$ 
is specified with $p_{X}$.
We next define the key used in a secret-key 
cryptosystem. Let $K$ be a random variable 
from the same finite set $\mathcal{X}$, having the distribution $p_{K}=\{p_K(k)\}_{k \in \mathcal{X}}$.
Let $\{K_t\}_{t=1}^\infty$ 
be a stationary DMS specified with $p_K$.

\noindent 
\underline{\it Random Variables and Sequences:} \/We write 
the sequence of random variables of length $n$ 
from the information source as follows 
$\bm{X}\coloneqq X_1X_2\ldots X_n$. 
Similarly, the strings of length $n$ 
in $\mathcal{X}^n$ are written 
as $\bm{x} \coloneqq x_1\cdots x_n\in\mathcal{X}^n$. 
For ${\bm{x}} \in \mathcal{X}^n$, 
$p_{\bm{X}}(\bm{x})$ stands for the 
probability of the occurrence of $\bm{x}$. 
When the information source is memoryless 
and specified by $p_X$, 
we have $p_{\bm{X}}(\bm{x})=
\prod_{t=1}^n p_X(x_{t})$. 
In this case, we write $p_{\bm{X}}(\bm{x})$ 
as $p_X^n(\bm{x})$. 
Similar notations are used for other random 
variables and sequences.

\noindent \underline{\textit{Other Notation:}}~Without 
loss of generality, we assume that $\mathcal{X}$ 
is a finite field.
$\oplus$ denotes addition over the field, and $\ominus$ 
denotes subtraction over the field.
As an example, for any $a,b$ in the same finite field,
we have $a \ominus b = a \oplus (-b)$.
Moreover, in this paper, all logarithms are 
taken to base 2, and $\mathbb{N}$ denotes the 
set of natural numbers.

\subsection{Basic System Description}
\noindent 
\underline{\it Source Coding without Encryption}:\/ We 
consider the variable length source coding without encryption. Let a sequence 
output from the source be denoted by $\bm{x} \in \mathcal{X}^n$ and let its encoded representation be denoted by $\underline{y}$.
Let $\displaystyle \mathcal{Y}^* := 
\bigcup_{l \in \mathbb{N}} \mathcal{Y}^l$
be the set consisting of all finite-length 
sequences composed of symbols in $\mathcal{Y}$.
Then $\underline{y} \in \mathcal{Y}^*$.
The source coding without encryption consists 
of the following three processes:    
\begin{figure}[t]
  \centering
  \includegraphics[width=60mm]{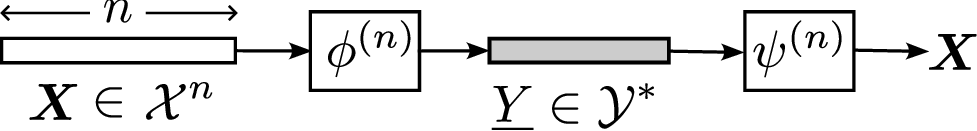}
  \caption{Source coding with variable-length codes} %
  \label{fig:SourceEnc}
 \end{figure}
\begin{itemize}
  \item[1)]{\it Encoding process:}\/ The sequence $\bm{x}$ 
  is encoded using an encoder $\phin$ defined by 
  $\phin: \mathcal{X}^n \to \mathcal{Y}^*$.
  \item[2)]{\it Transmission:}\/ The encoded sequence 
  $\underline{y}=\phin({\bm{x}})$ is transmitted 
  through a noiseless channel.
  \item[3)]{\it Decoding process:}\/ 
  The decoder $\psin$ defined by $\psin:\mathcal{Y}^* \to \mathcal{X}^n$ receives 
 $\underline{y}=\phin({\bm{x}})$ to decode the sequence $\bm{x}$ of the source output.    
\end{itemize}
The above processes of the source coding is 
shown in Fig. \ref{fig:SourceEnc}.
Throughout this paper we assume $\mathcal{Y}=\{0,1\}$.
Let $|\phin (\bm{x})|$ represents the length of the codeword 
$\phin (\bm{x}) \in \mathcal{Y}^*$. Define $l_{\max}$ by 
\begin{align*}
  l_{\max} \coloneqq \max_{ \bm{x} \in \mathcal{X}^n} 
  |\phin(\bm{x})|.
\end{align*}  
For each $1\leq l\leq l_{\max}$, define  
$$
\mathcal{D}_l \coloneq \{\bm{x}:~\psin \circ \phinx = \bm{x},~|\phin (\bm{x})| = l \},
 $$ 
which 
represents a set of all sequences in $\mathcal{X}^n$ which are mapped to codewords of length $l$ and are correctly decoded.
We set
$\underline{\cal D}\coloneqq$ 
$\{\mathcal{D}_l\}_{l=1}^{l_{\max}}$ and call it 
the uniquely decodable set. 
\newcommand{\OmiTTFig}{
\begin{figure}[t]
  \centering
  \includegraphics[width=70mm]{D_shazou.eps}
  \caption{Relationship between 
  $(\phin,\psin)$ and 
  $\mathcal{D}_l,1\leq l\leq l_{\max}$.}\label{fig:D_shazou}
\end{figure}
}
%
By the above definition 
we have the following: 
  \begin{align*}
      \mathcal{D}_l \cap \mathcal{D}_{l'}                         &= \emptyset,~1\leq l \neq l'\leq l_{\max}, 
      \:\: \bigcup_{l=1}^{l_{\max}} 
      \mathcal{D}_l= \mathcal{X}^n.
  \end{align*}
Define the random variable 
$L^{\prime}=L^{\prime}_{\phin}(\bm{X})$ 
representing the codeword length
of $\phin(\bm{X})$ and its expectation 
$\overline{L}^{\prime}_{\phin}$ 
by \begin{align*}
L^{\prime}_{\phin}(\bm{X})
\coloneqq |\phinX|,\quad 
\overline{L}^{\prime}_{\phin} 
\coloneqq \expectation[L^{\prime}_{\phin}(\bm{X})].
\end{align*}
\newcommand{\OmiTTLemma}{
For source coding systems without encryption, 
we have the following lemma.
\begin{lemma} \label{lem:scs_tilde}
  For a given $(\phin,\psin)$, we have 
  the following:
  \begin{itemize}
    \item [a)] $|\mathcal{D}_l| 
      \leq 2^l,~l=1,2,\cdots, l_{\max}$,
    \item [b)] 
    There exists $(\widetilde{\phi}^{(n)},
    \widetilde{\psi}^{(n)})$ such that
 \begin{align*}
& \overline{L}'_{\widetilde{\phi}^{(n)}} 
  \leq \overline{L}^{\prime}_{\phin}+1.
\\
& \overline{L}'_{\widetilde{\phi}^{(n)}} 
  \leq \sum_{l=1}^{l_{\max}}p(\mathcal{D}_l)   
  \log |\mathcal{D}_l|+1,
\\
&\overline{L}'_{\widetilde{\phi}^{(n)}} 
 \geq \sum_{l=1}^{l_{\max}}p(\mathcal{D}_l) 
 \log |\mathcal{D}_l|. 
\end{align*}
  \end{itemize}
\end{lemma}

Proof of Lemma \ref{lem:scs_tilde} is given in \ref{prf:scs_tilde}.
}
\noindent 

\newcommand{\FFschme}{

}

\noindent
\underline{\it Source Encryption:}\/ We consider 
the variable-length source encryption.
We assume that the sequence and the key are independently generated from $\mathcal{S}_{\mathsf{gen}}$ and $\mathcal{K}_{\mathsf{gen}}$, respectively.
The sequence $\bm{X}$ generated from $\mathcal{S}_{\mathsf{gen}}$ and 
the key ${\bm{K}}$ generated from 
$\mathcal{K}_{\mathsf{gen}}$ 
are both sent 
to node $\mathsf{L}$. The details of the source 
encryption system are shown in Fig.~\ref{SourceEnc_2}.
\begin{figure}[t]
  \centering
  \includegraphics[width=60mm
  ]{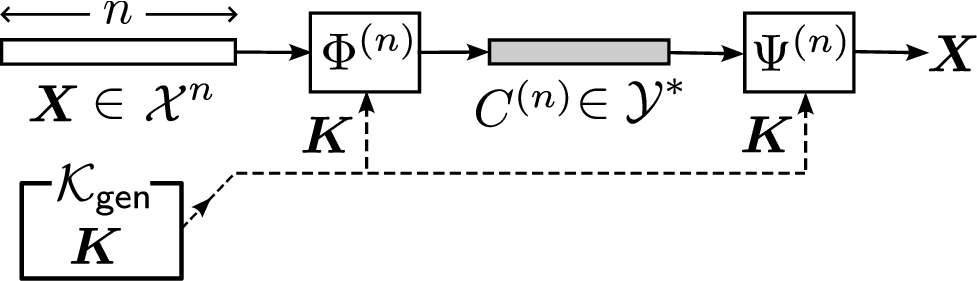}
  \caption{Source encryption with variable-length 
  codes 
  }\label{SourceEnc_2}
  \vspace*{-3mm}
\end{figure}
\begin{itemize}
  \item[1)]{\it Source Processing:} \/At node $\mathsf{L}$, $\bm{X}$ is encrypted with the key $\bm{K}$ using the encryption function
  $\Phin : \mathcal{X}^n \times \mathcal{X}^n \to \mathcal{Y}^*$.
  The ciphertext of $\bm{X}$ is given by $C^{(n)} = \Phin(\bm{K}, \bm{X})$.
  On the encryption function $\Phin$, we use the following notation:
  $ \Phin(\bm{K}, \bm{X})
    = \PhiKnX.$

\item[2)]{\it Transmission:} The ciphertext $C^{(n)}$ is sent to node $\mathsf{D}$ through the public communication channel. Meanwhile, the key $\bm{K}$ is sent to $\mathsf{D}$ through the private communication channel.
\item[3)] {\it Sink Node Processing:} \/
  At node $\mathsf{D}$, the ciphertext is decrypted 
  using the key $\bm{K}$ through the corresponding 
  decryption procedure
  $\Psin : \mathcal{X}^n \times 
  \mathcal{Y}^* \to \mathcal{X}^n$.
On the decryption function $\Psin$, we use the 
following notation:
$ \Psin(\bm{K}, C^{(n)})= \PsiKn(C^{(n)})$.
\end{itemize}
We fix $\bm{K}=\bm{k} \in \mathcal{X}^n$ arbitrary. 
For each $\bm{k}\in \mathcal{X}^n$, 
define 
\begin{align*}
\mathcal{D}_{l,\bm{k}}
\coloneq\{\bm{x}:~\Psi_{\bm{k}}^{(n)}\circ \Phiknx
=\bm{x},~|\Phi_{\bm{k}}^{(n)} (\bm{x})|=l\},
\end{align*}
which represents the set of sequences that 
are encoded into codewords of length $l$. 
For each $\bm{k}\in {\cal X}^n$, we set
$\underline{\cal D}_{\bm{k}}\coloneqq$ 
$\{\mathcal{D}_{l,\bm{k}}\}_{l=1}^{l_{\max}}$ 
and call it 
the uniquely decodable set for given 
$\bm{k}\in {\cal X}^n$.
It is obvious that 
%
$\forall \bm{k} \in \mathcal{X}^{n}$ and 
$\forall \bm{x} \in \mathcal{X}^{n}$,  
$ \Psi_{\bm{k}}^{(n)} \circ \Phiknx 
  =  \bm{x}$,
which implies that  
for each $\bm{k}\in {\cal X}^n$ and each 
$n\in \mathbb{N}$, 
$\Phi_{\bm{k}}^{(n)}:{\cal X}^n\to {\cal Y}^*$ 
is injective. For each $\bm{k}\in \mathcal{X}^n$, 
define the random variable 
$L = L_{\Phi_{\bm{k}}^{(n)}}(\bm{X})$ by
$
L_{\Phi_{\bm{k}}^{(n)}}($ $\bm{X}) \coloneqq |\PhiknX|.
$


\subsection{Average Codeword Length 
of Source Coding Systems}
For source coding without encryption, 
the average codeword length denoted by  $\overline{L}^{\prime}_{\phin}$ is 
\begin{align*}
  \overline{L}'_{\phin} 
  = \expectation[L'_{\phin}(\bm{X})] 
  = \expectation[|\phinX|].
\end{align*}
For source encryption, the average 
codeword length $\overline{L}_{\Phi_{\bm{k}}^{(n)}}$ for an arbitrary key 
$\bm{k} \in \mathcal{X}^n$ is 
\begin{align*}
 \overline{L}_{\Phi_{\bm{k}}^{(n)}} 
  =\expectation\left[
 {L}_{\Phi_{\bm{k}}^{(n)}}(\bm{X})
 \right]
 = \expectation\left[\left|\PhiknX\right|\right].
\end{align*}
We define
\begin{align}
 \overline{L}_{\Phin} 
  :=& 
  \expectation\left[
  \overline{L}_{\Phi_{\bm{K}}^{(n)}} 
\right]
 = \expectation
  \left[{L}_{\PhiKn}(\bm{X})
 \right]
\notag \\ 
 =& \sum_{\bm{k} \in \mathcal{X}^n} p_K^n(\bm{k})\expectation\left[
{L}_{\Phi_{\bm{k}}^{(n)}}(\bm{X}) \right].
\label{eqn:LtildeLaOne}
\end{align}
For the proposed source coding system 
with encryption, we have the 
following proposition:
\begin{proposition} \label{pos:Ltilde}
 We have the following two bounds:
 \begin{align}
&
\overline{L}_{\Phin}
\geq nH(X)- \log n - 
\log(2e\log |\mathcal{X}|). 
 \label{eqn:LtildeLbTwo}   
 \end{align}
 Here in the third term in the right members of  
 (\ref{eqn:LtildeLbTwo}), $e=2.718\ldots$ 
 is Nepier's constant.   
 \end{proposition}

Proof of Proposition \ref{pos:Ltilde} 
is given in Appendix \ref{prf:pos_Ltilde}.
This proposition plays an important role 
in establishing the converse coding theorem.
\subsection{Problem Set Up}
In this subsection, we introduce a security criterion and reliability criterion. 

\noindent \underline{\it Security Criterion:} \/The adversary $\mathcal{A}$ attempts to estimate the source sequence $\bm{X} \in \mathcal{X}^n$ from the ciphertext $C^{(n)}$.
We therefore define the mutual 
information (MI) between $\bm{X}$ and $C^{(n)}$ as
$\Delta_{\mathrm{MI}}^{(n)}
$ $=\Delta_{\mathrm{MI}}^{(n)}(\Phin\mid p_X^n, p_K^n)
\coloneqq I(C^{(n)}; \bm{X}),$
and adopt it as the security criterion.

\begin{definition} \label{def:e-anzen_fv}
    We fix some positive constant $\delta_0$. 
    For a fixed $\delta \in (0,\delta_0]$,
    a quantity $R$ is $\delta$-admissible if $\exists \{(\Phin,$ $\Psin)\}^{\infty}_{n=1}$ such that
    $\forall \gamma >0$, 
    $\exists n_0=n_0(\gamma) \in \mathbb{N}$, 
    $\forall n\geq n_0$, 
  \begin{align*}
    &\frac{1}{n} \overline{L}_{\Phin} = \frac{1}{n} \expectation\left[L_{\PhiKn}(\bm{X})\right] 
    \leq R + \gamma, \\
    &\Delta_{\mathrm{MI}}^{(n)}
      (\Phin|{p_X^n},p_K^n)=I(C^{(n)};\bm{X})\leq \delta.
  \end{align*}
\end{definition}
\begin{definition}
  Let $R_{\mathrm{v,inf}}(\delta |p_X, p_K)$ denote the infimum of all $\delta$-admissible rates $R$.
  Furthermore, we define
  \begin{align*}
    R_{\mathrm{v,inf}}({p}_X, {p}_K) \coloneqq \sup_{\substack{\delta \in (0,\delta_0]}} R_{\mathrm{v,inf}}(\delta|{p}_X, {p}_K).
  \end{align*}
 Since $R_{\mathrm{v,inf}}(\delta|{p}_X,{p}_K)$ 
is monotone decreasing with respect to $\delta\in(0,\delta_0]$, 
$R_{\mathrm{v,inf}}({p}_X, {p}_K)$ 
does not depend on $\delta_0$.
\end{definition}

\section{Main Results}

Let $\mathcal{P}(\mathcal{X})$ denote 
the set of all probability distributions 
on $\mathcal{X}$. For $R \geq 0$ 
and $p_X \in \mathcal{P}(\mathcal{X})$, 
we define 
\begin{align*}
 & E_{\empty}(R|p_X) \coloneqq 
  \min_{\substack{ P \in \mathcal{P}(\mathcal{X}):  \\ 
  R  \leq H(P)}} D(P ||p_X),
 \\
 & F(R|p_K) \coloneqq 
 \min_{P \in \mathcal{P}(\mathcal{X})}
 \left\{
 [H(P)-R]^{+} + D(P||p_K)
 \right\}.
\end{align*}
Here $[a]^{+}\coloneq \max\{0,a\}$.
For the functions $E_{\empty}(R|p_X)$ and $F(R|p_K)$, 
we have the following property.
\begin{property}\label{per:FRpKFV} 
The two functions $E_{\empty}(R|p_X)$ and $F(R|p_K)$ take positive values
if and only if $H(X) \! <\! R\! <\! H(K)$.
\end{property}

We set 
\begin{align}
  \gamma_n &\coloneqq 
  \frac{1}{n}\{|\mathcal{X}|\log(n+1)
  +\log|{\cal X}|\},\:
  R_n 
  \coloneq R+\gamma_n.
 \label{eqn:ChoicRate}  
\end{align}
Note that $\gamma_n$ vanishes  as $n \to \infty$.
On the direct coding
theorem we have the following result:  
\begin{theorem}[Direct Coding Theorem]
\label{MI_FV}
$\forall R >0$, 
$\forall n\in\mathbb{N}$, 
and $\forall (p_X,p_K)  
\in {\cal P}^2({\cal X})$,
\begin{align}
 &\frac{1}{n} \overline{L}_{\Phin}
 =\frac{1}{n} 
 \expectation[L_{\PhiKn}(\bm{X})]
 \notag\\   
 &\leq R_n+\frac{1}{n}+(\log |\mathcal{X}|-R)
     (n+1)^{|\mathcal{X}|}
     2^{-nE_{\empty}(R|p_K)},
 \label{upb:avgL} \\     
&\Delta_{\mathrm{MI}}^{(n)} (\Phin|{p_X^n},p_K^n)
 =I(C^{(n)};\bm{X})
\notag\\
 & \leq \left\{
       n[E_{\empty}(R|p_X) +  \log |\mathcal{X}|]
        +\log e \right\}(n+1)^{|\mathcal{X}|}
 \notag\\
     &\qquad \times 2^{-nE_{\empty}(R|p_X)}
      +(R_n+1)|{\cal X}|(n+1)^{4|{\cal X}|}2^{-nF(R|p_K)}.
      \label{upb:vMI}
  \end{align}
\end{theorem}

Proof of Theorem \ref{MI_FV} is given 
in Section \ref{subsec:ThmDirect}.
From Theorem~\ref{MI_FV} 
and Property \ref{per:FRpKFV}, we have 
the following corollary: 
\begin{corollary}\label{cor:directThA}
$\forall R>0$,
$\exists \{(\Phin,\Psin)$$\}_{n=1}^{\infty}$ 
such that 
$\forall (p_X,$ $p_K)\in {\cal P}^2({\cal X})$ 
with $H(X)<R< H(K)$, 
\begin{align}
\left.
\begin{array}{l}
{\displaystyle \limsup_{n\to\infty}} 
 (1/n)\overline{L}_{\Phin}\leq R, 
\vspace{1mm}\\
{\displaystyle \liminf_{n\to\infty}}
(-1/n) 
\log \MI(\Phin|{p_X^n},p_K^n)
\\ 
\geq \min\{E(R|p_X), F(R|p_K)\}>0.
\end{array}
\right\}
\label{eqn:ThOneBound}
\end{align}
\end{corollary}
By Corollary \ref{cor:directThA}, under 
$H(X) <R<H(K)$,  
we have the followings: 
\begin{itemize}
\item[1.] 
The average 
length of codewords per symbol 
$(1/n)\overline{L}_{\Phin}$ 
is asymptotically
upper bounded by $R$.
\item[2.] On the security, 
$\Delta_{\rm MI}^{(n)}
(\Phi^{(n)}$$|p_{X}^n,p_{K}^n)$ vanishes exponentially 
as $n\to \infty$, and its exponent 
is lower bounded by $\min\{E_{\empty}(R|p_X),
                           F(R|p_{K})\}$.
\item[3.] The code that attains 
the exponent function $\min\{E_{\empty}($ $R|p_X),
                           F(R|p_{K})\}$
is the universal code not depending on 
$(p_{X},p_K)\in {\cal P}^2({\cal X})$.
\end{itemize}

\noindent

We define the following quantity.
\begin{align*}\label{def:Rast2}
   R^{\ast}(p_X, p_K) 
   =
  \begin{cases}
    H(X) &\mbox{ if } H(X) < H(K), \\
    +\infty    &\mbox{ otherwise}.
  \end{cases}
  \stepcounter{equation}\tag{\theequation}
\end{align*}
From the definition of $R^{\ast}(p_X, p_K)$, 
Corollary \ref{cor:directThA}, 
and Property~\ref{per:FRpKFV}, 
we have the following corollary: 
\begin{corollary} \label{cor:directTh_FV}
  For $\delta \in (0,\delta_0]$, we have
  \begin{align*}
    R_{\mathrm{v,inf}}(\delta|p_X,p_K) 
    \leq R_{\mathrm{v,inf}}(p_X,p_K) 
    \leq R^{\ast}(p_X,p_K).
  \end{align*}
\end{corollary}

\textit{Proof:}
Choose $R$ such that $H(X) < R < H(K)$. 
Since $H(X)<H(K)$, 
this choice of $R>0$ is possible.   
Then, for any $\tau \in \bigl(0,\min{R-[H(X), H(K)-R}\bigr]$, we have the following inequality:
\begin{align}\label{AA_FV}
H(X) + \tau \leq R \leq H(K) - \tau.
\end{align}
By Property \ref{per:FRpKFV}, for any $R$ satisfying \eqref{AA_FV}, both $E_{\empty}(R|p_X)$ and $F(R|p_K)$ take positive values.
Then it follows from 
Corollary \ref{cor:directThA}
that $\exists 
\{(\Phin,\Psin)\}_{n =1}^{\infty}$ 
such that we have 
the bound (\ref{eqn:ThOneBound}) 
in this corollary, implying 
that for any $\delta$ 
$\in (0,\delta_0]$, 
every $R$ satisfying \eqref{AA_FV} 
is $\delta$-admissible.
Since $\tau > 0$ 
in \eqref{AA_FV} 
can be chosen arbitrarily small, it follows 
that under $H(X)<H(K)$, every $R$ satisfying 
$H(X)\leq R\leq H(K)$
is $\delta$-admissible.
Combining this fact with the definition of 
$R^{\ast}(p_X,p_K)$ given by \eqref{def:Rast2}, 
we have that 
for any $\delta$ 
$\in (0,\delta_0]$, 
$R^{\ast}(p_X,p_K)$ is 
$\delta$-admissible. 
Since $\delta$ can arbitrary 
be close to $0$, we conclude 
that $R^{\ast}(p_X,p_K)$ 
$\geq R_{\rm v,inf}(p_X,p_K)$.   
\hfill\IEEEQED

We next describe a result on the 
converse coding theorem. 
To this end we set    
\begin{align*}\label{def:Rast3}
   R^{\star}(p_X, p_K) 
   =
  \begin{cases}
    H(X) &\mbox{ if } H(X) \leq H(K), \\
    +\infty         &\mbox{ otherwise}.
  \end{cases}
  \stepcounter{equation}\tag{\theequation}
\end{align*}
On a lower bound of
$R_{\mathrm{v,inf}}(\delta |{p}_X, {p}_K), \delta \in (0,\delta_0]$, we have the following result:  
\begin{theorem}[Converse Coding Theorem] 
\label{th:converseTh_FV}
  For $\delta \in (0,\delta_0]$, 
   $ R^{\star}(p_X, p_K) 
    \leq R_{\mathrm{v,inf}}(\delta|p_X, p_K) \leq R_{\mathrm{v,inf}}(p_X, p_K).$ 
\end{theorem}

Proof of Theorem \ref{th:converseTh_FV}  
is given in Section \ref{subsec:ThmConv}.
From Corollary \ref{cor:directTh_FV} 
and Theorem \ref{th:converseTh_FV}, 
we obtain the following theorem:
\begin{theorem} \label{th:codingTh_FV}
Consider the case of $H(X)<H(K)$. In this case 
we have  
$
R^{\star}(p_X, p_K)=R^{\ast}(p_X, p_K)=H(X).
$
Hence for any $\delta \in (0,\delta_0]$, 
    $R_{\mathrm{v,inf}}(p_X,p_K) 
    = R_{\mathrm{v,inf}}(\delta|p_X,p_K)=H(X).$
\end{theorem}

Theorem \ref{th:codingTh_FV} implies the 
strong converse property for 
$R_{\mathrm{v,inf}}($ $\delta |{p}_X, {p}_K),
\delta \in (0,\delta_0]$.

\section{Proofs of Theorems 
\ref{MI_FV} and \ref{th:converseTh_FV}}
\subsection{Proof of Theorem~\ref{MI_FV}}
\label{subsec:ThmDirect}
In this subsection, we give the proof of 
Theorem~\ref{cor:directTh_FV}. 
We propose a source encryption 
scheme to prove Theorem~\ref{MI_FV}.
We first state some choices of parameters 
we set for the construction of source coding and 
encryption scheme.      

\noindent
\underline{\it Choices of Some Parameters:}\/ 
We choose $m\in \mathbb{N}$ such that
\begin{align}
  m = 
  \left\lfloor\frac{nR_n}{\log 
  |\mathcal{X}|}\right\rfloor, \: R_n=R+\gamma_n.
\label{eqn:ChoosEm}
\end{align}
where $\lfloor a \rfloor$ stands for the integer 
part of $a\in \mathbb{R}$. 
The choice (\ref{eqn:ChoosEm}) of $m$ implies 
the following: 
\begin{align}\label{Rconstraint}
 \frac{m}{n} \log |\mathcal{X}|\leq R_n  
 \leq \frac{m+1}{n}\log{|\mathcal{X}|}. 
\end{align}

We present two propositions necessary for the proof 
of Theorem~\ref{MI_FV}. 
Under the choice (\ref{eqn:ChoosEm}) of $m$, we have 
the following two propositions:  

\begin{proposition} \label{pos:Kobayashi}
$\exists \{ \mathcal{C}_{\empty}^{n}(R)\}_{n=1}^{\infty}$ 
with $\mathcal{C}_{\empty}^{n}(R) 
    \subseteq \mathcal{X}^n$ such 
that $\forall n \in \mathbb{N}$ 
and $\forall p_X\in {\cal P}({\cal X})$, 
\begin{align}
& |\mathcal{C}^{n}(R)| 
  \leq (n+1)^{|\mathcal{X}|} 2^{nR},
 \label{eqn:calCupb} \\
& \Pr\{{\bm X}\in \mathcal{X}^n 
                 -\mathcal{C}_{\empty}^{n}(R)\} 
  \leq (n+1)^{|\mathcal{X}|}2^{-nE_{\empty}(R|p_X)}.
 \label{eqn:ErrorUpb}
\end{align}
\end{proposition}

\begin{proposition}\label{pos:upb_Max-MI}
$\exists\{\varphi^{(n)}\}_{n=1}^{\infty}$ with $\varphi^{(n)}:\mathcal{X}^n \to \mathcal{X}^m$
such that 
 $\forall n \in \mathbb{N}
 $ 
and $\forall p_K\in {\cal P}({\cal X})$, 
\begin{align*}
& m \log |\mathcal{X}| 
- H(\widetilde{K}^m)
\leq \left(R_n+\frac{1}{2}\right)(n+1)^{3|{\cal X}|}
      2^{-n[F(R|p_K)-\gamma_n]}.
\end{align*}
Here we set $\widetilde{K}^m=\varphi^{(n)}(\bm{K})$.  
\end{proposition}

Proofs of Propositions \ref{pos:Kobayashi} 
and  \ref{pos:upb_Max-MI}
are found in 
\cite{oohamaSa26FrwOfSEncUnivStConv
}.
We use Proposition \ref{pos:Kobayashi} 
for the construction of $(\phin,$ $\psin)$. 
We use Proposition \ref{pos:upb_Max-MI}
to evaluate $\Delta_{\rm MI}^{(n)}
(\Phin|{p_X^n},p_K^n)=I($ $C^{(n)};\bm{X})$ 
for the proposed source encryption scheme.    

\noindent 
\underline{\textit{Construction of 
$(\wt{\phi}^{(n)},\wt{\psi}^{(n)})$:}} \/ 
We first derive upper bounds of 
$|\mathcal{C}_{\empty}^{n}(R)|$ 
under the cardinality bound (\ref{eqn:calCupb})
in Proposition \ref{pos:Kobayashi}.
We have 
the following chain of inequalities:
\begin{align}
&|\mathcal{C}_{\empty}^{n}(R)|
\MLeq{a} (n+1)^{|\mathcal{X}|} 2^{nR}
 = 2^{-n\gamma_n}(n+1)^{|\mathcal{X}|}2^{nR_n}
\notag\\
& \MLeq{b} 2^{-n\gamma_n }(n+1)^{|\mathcal{X}|}
  |{\cal X}|^{m+1}
\MEq{c}|{\cal X}^m|\label{eqn:upbCR}.  
\end{align}
Step (a) follows from 
the bound (\ref{eqn:calCupb}) in 
Proposition \ref{pos:Kobayashi}. 
Step (b) follows from that 
the bound (\ref{Rconstraint})
implies $2^{nR_n} \leq |\mathcal{X}|^{m+1}$. 
Step (c) follows from 
$2^{-n\gamma_n} = 
\left[(n+1)^{|\mathcal{X}|}
|\mathcal{X}|\right]^{-1}$.
The cardinality bound (\ref{eqn:upbCR}) 
implies that there exists a one-to-one mapping 
$\widetilde{\phi}^{(n)}_0:\mathcal{C}_{\empty}^{n}(R) 
\to \mathcal{X}^m$.
Using $\widetilde{\phi}_0^{(n)}$, we define the mapping 
$\widetilde{\phi}^{(n)}:\mathcal{X}^n \to \mathcal{X}^m$ 
by 
\begin{align*}
  \widetilde{\phi}^{(n)}
  \coloneqq\left\{\begin{array}{cl}
    \widetilde{\phi}_0^{(n)}(\bm{x}),
     &\mbox{ if }\bm{x} 
     \in \mathcal{C}_{\empty}^{n}(R),
    \vspace{0.2cm}\\
    \bm{x},
     &\mbox{ otherwise}.
  \end{array}
  \right.
\end{align*}
We define $\wt{\psi}^{(n)}:{\cal X}^m\to {\cal X}^n$ based on the one-to-one mapping 
$\widetilde{\phi}_0^{(n)}:\mathcal{C}_{\empty}^{n}(R) \to 
\mathcal{X}^{m}$. Concretely we define $\wt{\psi}^{(n)}$ 
so that for each 
$\widetilde{x}^m $ $\in $ $\widetilde{\phi}^{(n)}(\mathcal{C}_{\empty}^{n}(R))$, 
$\wt{\psi}^{(n)}(\widetilde{x}^m)
=\bm{x}$. Here $\bm{x}$ is the unique element 
such that $\wt{\phi}_0^{(n)}(\bm{x})
=\widetilde{x}^m$.
For $\bm{x}\in \wt{\phi}^{(n)}(\mathcal{X}^n -
\mathcal{C}_{\empty}^{n}(R))$
$=\mathcal{X}^n-\mathcal{C}_{\empty}^{n}(R) $, 
we define $\wt{\psi}^{(n)}$ such that 
$\wt{\psi}^{(n)}(\bm{x})=\bm{x}$.
The encoding and decoding procedures 
is shown in Fig.~\ref{fig:FV_CRn_shazou}.
\begin{figure}[t]
  \centering
  \includegraphics[width=84mm]{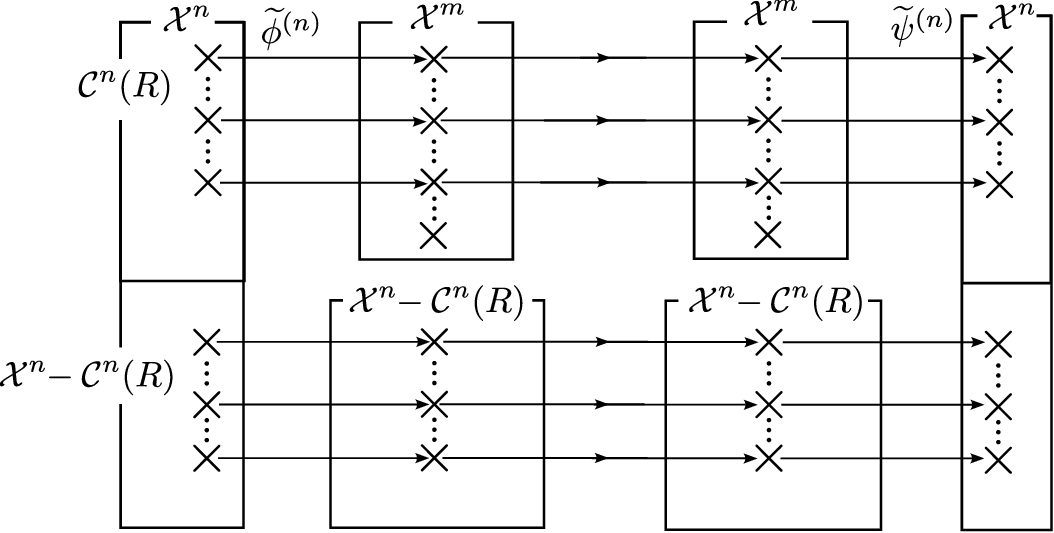}
  \caption{
  Encoding and decoding procedures} 
  \label{fig:FV_CRn_shazou}
\end{figure}

\noindent 
\underline{\it Source Encryption Scheme:} \/ 
To concretely construct $({\Phi}^{(n)},$ ${\Psi}^{(n)})$, we provide 
some definitions. Let 
$\wtPhin: 
{\cal X}^{n}\times {\cal X}^n 
\to {\cal X}^{m}\cup {\cal X}^n$ and 
$\wtPsin: {\cal X}^{n} \times 
({\cal X}^{m}\cup {\cal X}^n)  
\to {\cal X}^n $
be two maps, whose constructions 
will be stated later. 
Let  $\lceil a \rceil$ stands for 
the smallest integer not below $a$. 
Let $\nu_1:{\cal X}^{m} \to 
\{0,1\}^{\lceil nR_n \rceil}$  
and $\nu_2:{\cal X}^{n} \to 
\{0,1\}^{\lceil n \log |{\cal X}|\rceil}$ 
be two arbitrary injective maps.
Set ${\cal B}_1\coloneq\nu_1({\cal X}^m)$ and 
    ${\cal B}_2\coloneq \nu_2({\cal X}^n)$. 
Let $\nu: {\cal X}^m\cup {\cal X}^n \to B_1 \cup B_2$
be a bijective map such that
\newcommand{\NuMap}{
}{
\begin{align*}
\nu(\underline{a})=\left\{
 \begin{array}{l}
 \nu_1(\underline{a})\mbox{ if } \underline{a} 
 \in {\cal X}^{m},
 \vspace*{1mm}\\
 \nu_2(\underline{a})\mbox{ if } 
 \underline{a}\in {\cal X}^{n}.
 \end{array}\right.
  \end{align*}
}
Let $\kappa: B_1\cup B_2\to  {\cal X}^m \cup {\cal X}^n$
be a bijective map such that for each $i=1,2$ and for 
each $\underline{b}\in {\cal B}_i$, 
$ \kappa(\underline{b})=\nu_i^{-1}(\underline{b})$. 
It is obvious that $\kappa=\nu^{-1}$.   
The source encryption scheme
$({\Phi}^{(n)},$ ${\Psi}^{(n)})$ 
consists of the following three steps:
\newcommand{\FigBinEx}{
\begin{figure}[t]
  \centering
  \includegraphics[width=88mm]{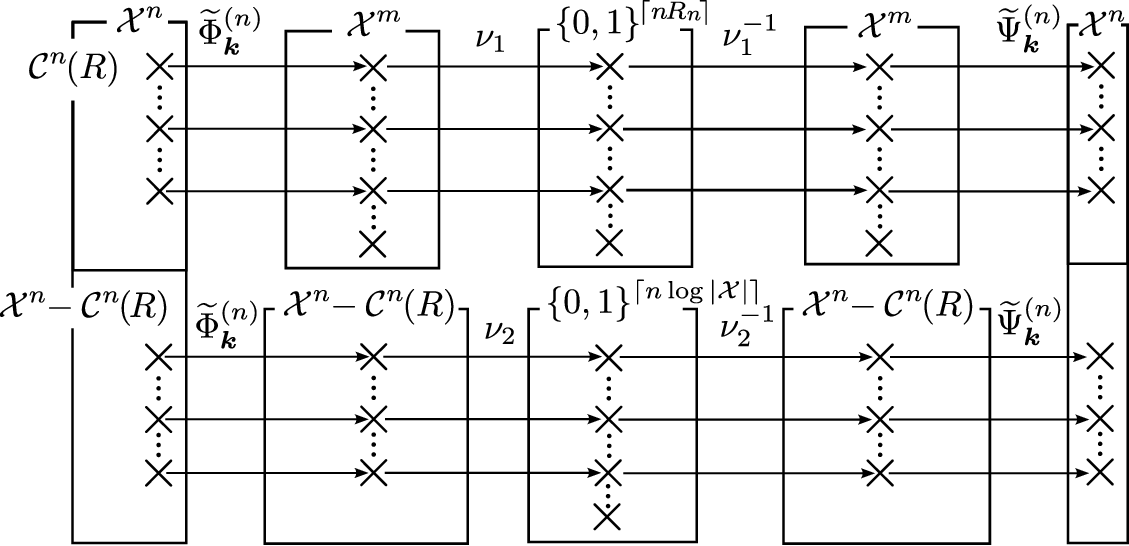}
  \caption{
  Binary sequence expressions of ciphertexts} 
  \label{fig:FV_BinaryExPr}
\end{figure}}
\newcommand{\FVScEncSh}{
\begin{figure}[t]
 \centering
 \includegraphics[width=88mm]{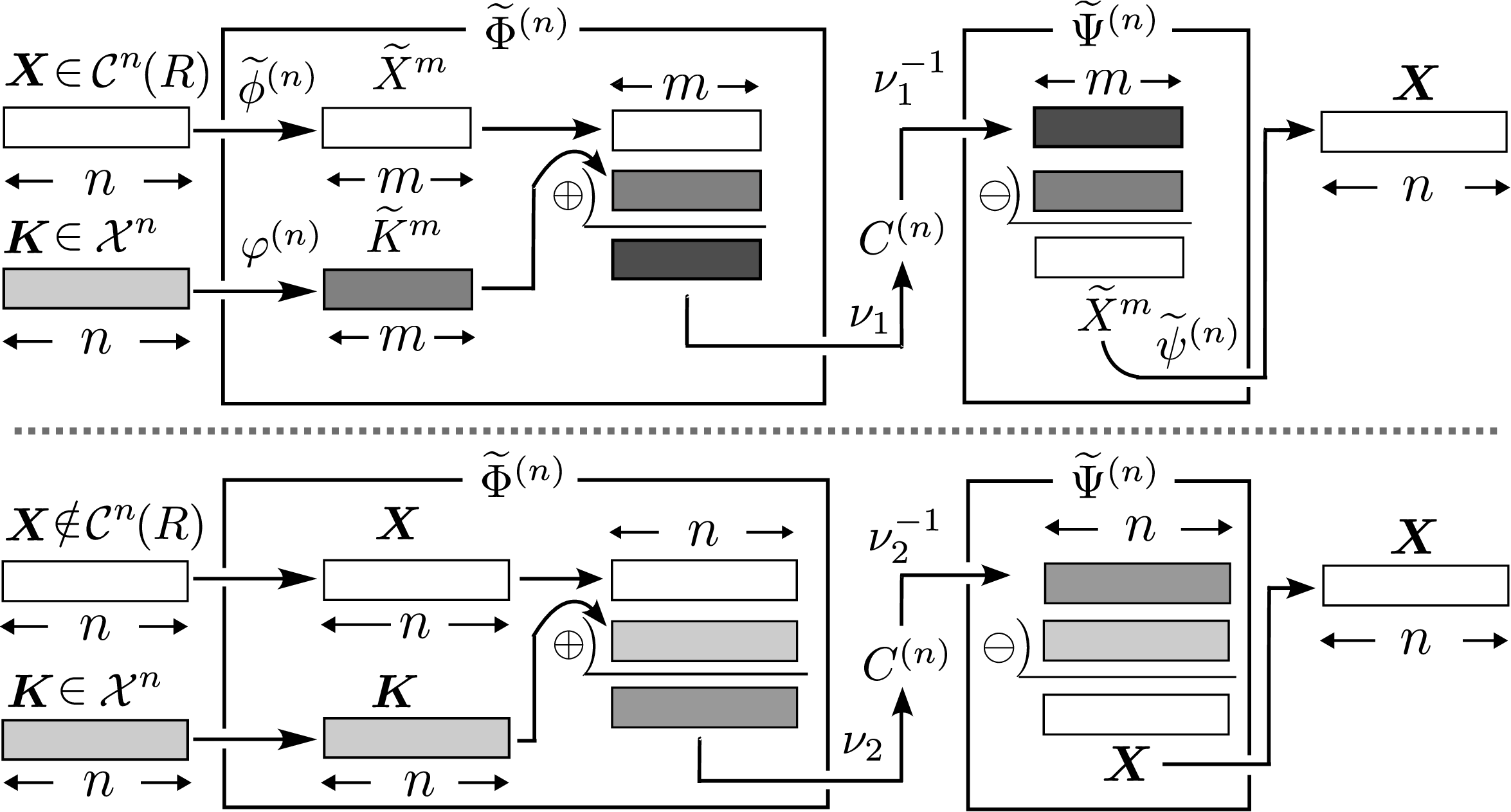}
 \caption{Encryption and decryption procedures}
 \label{FV_SourceEncScheme}
\end{figure}}
\FigBinEx
\begin{itemize}
  \item[1)] {\it Construction of 
  $\widetilde{\Phi}^{(n)}$:} \/ 
  For $(\bm{k},\bm{x})\in {\cal X}^n 
  \times {\cal X}^n$, 
  define  $\widetilde{\Phi}^{(n)}$
  by
  \begin{align*}
   &  \widetilde{\Phi}^{(n)}(\bm{k},\bm{x})
   =\wtPhiknx
      \\
   & 
   \coloneqq\left\{\begin{array}{ll}
      \hspace*{-3pt}\varphi^{(n)}(\bm{k}) 
      \oplus \widetilde{\phi}^{(n)}(\bm{x}),
     &\hspace*{-4pt}\mbox{if }\bm{x} 
          \in \mathcal{C}_{\empty}^{n}(R),
          \vspace{0.2cm}\\
           \hspace*{-4pt} 
           \bm{k}\oplus \bm{x}, 
            &\hspace*{-3pt}\mbox{if }\bm{x}\in \mathcal{X}^n
            -\mathcal{C}_{\empty}^{n}(R).
          \end{array}
          \right.
  \end{align*}
 \item[2)]{\it Binary Sequence Expressions 
  of Ciphertexts:}\/ 
  For each $\bm{k}\in {\cal X}^n$, binary 
  sequences 
  $\nu\circ\widetilde{\Phi}_{\bm{k}}^{(n)}(\bm{x})$, 
  $\bm{x}\in {\cal X}^n$ 
  are as follows.   
  If $\bm{x}\in {\cal C}_{\empty}^n(R)$, 
  $\nu_1\circ\widetilde{\Phi}_{\bm{k}}^{(n)}(\bm{x})
   $ is 
   a binary sequence with the length 
  $\lceil nR_n \rceil$. 
  If $\bm{x}\in \mathcal{X}^n
         -\mathcal{C}_{\empty}^{n}(R)$,  
  $\nu_2\circ\widetilde{\Phi}_{\bm{k}}^{(n)}(\bm{x})$
  is 
  a binary 
  sequence with the length 
  $\lceil n\log |{\cal X}|\rceil$. 
  This binary sequence is sent through 
  the public channel 
  (see Fig.~\ref{fig:FV_BinaryExPr}).
  We set 
   $
   C^{(n)}=\PhiKnX=
   \nu\circ\widetilde{\Phi}_{\bm{K}}^{(n)}(\bm{X}).
   $
   For the length of 
   $C^{(n)}=$ $\PhiKnX$, we have the following:
 \begin{align*}
   L_{\Phi_{\bm{K}}^{(n)}}(\bm{X})=
    \begin{cases}
    \lceil nR_n \rceil 
     &\mbox{ if }\bm{X}
       \in {\cal C}^n_{\empty}(R),
       \\            
      \lceil n \log |{\cal X}| \rceil 
      &\mbox{ if }\bm{X}\in \mathcal{X}^n
            -\mathcal{C}_{\empty}^{n}(R).
  \end{cases}  
  \end{align*}
  \item[3)] {\it Construction of $\Psin$:}\/ 
  $\Psin$ receives the ciphertext 
  $C^{(n)}=\PhiKnX$ and the key $\bm{K}$, respectively, 
  through public and private channels. For 
  $(\bm{K},C^{(n)})$, define ${\Psi}^{(n)}$ by
  \begin{align*}
  &\Psin(\bm{K},C^{(n)})
  =\wtPsin(\bm{K},\nu^{-1}(C^{(n)}))
  \\
  &=\left\{
  \begin{array}{rl}
   \wtPsin(\bm{K},\nu_1^{-1}(C^{(n)}))& 
   \mbox{if }|C^{(n)}|=  \lceil nR_n \rceil, \\
   \wtPsin (\bm{K}, \nu_2^{-1}(C^{(n)})) & 
   \mbox{if }|C^{(n)}|= 
    \lceil n\log |{\cal X}|\rceil.
  \end{array}
  \right.
  \end{align*}
The decryption process by  
$\wtPsin$ consists of the following 
three steps:  
\begin{itemize}
\item[i)] 
Using $\varphi^{(n)}$, ${\wtPsin}$ encodes 
$\bm{K}$ into $\widetilde{K}^m = \varphi^{(n)}(\bm{K})$.
\item[ii)] If $L_{\Phi_{\bm{K}}^{(n)}}(\bm{X})
     =\lceil nR_n \rceil$, 
     ${\wtPsin}$ 
subtracts $\widetilde{K}^m$ from 
$\nu_1^{-1}($ $C^{(n)})$ 
to obtain $\widetilde{X}^m = \phinX$.
Finally, $\wtPsin$ outputs ${\bm{X}}$ 
by applying the decoder $\wt{\psi}^{(n)}$ 
to $\widetilde{X}^m$.
\item[iii)]
 If $L_{\Phi_{\bm{K}}^{(n)}}(\bm{X})
     =\lceil n \log |{\cal X}| \rceil$, ${\wtPsin}$ 
subtracts $\bm{K}$ from 
$\nu_2^{-1}(C^{(n)})$ to obtain $\bm{X}$.
\end{itemize}
\end{itemize}
The above construction of  
$({\Phi}^{(n)},$ ${\Psi}^{(n)})$ is shown 
in Fig.~\ref{FV_SourceEncScheme}.
\FVScEncSh
Let $\phi^{(n)}\coloneqq \nu\circ\wt{\phi}^{(n)},
\psi^{(n)}\coloneqq \wt{\psi}^{(n)}\circ\nu^{-1}$.
By the constructions of 
$(\phi^{(n)},\psi^{(n)})$ and $(\Phin, \Psin)$, 
we have that $\forall \bm{k}\in {\cal X}^n$,  
\begin{align}
\underline{\cal D}=
\underline{\cal D}_{\bm{k}}
=\{{\cal D}_l\}_{l=\lceil nR_n\rceil, 
  \lceil n\log|{\cal X}|\rceil }.
  \label{eqn:LengthEq}
\end{align}
For the proposed $(\Phin, \Psin)$, 
we evaluate upper bounds of $\overline{L}_{\Phin}$ 
and $\Delta_{\rm MI}^{(n)}(\Phin|{p_X^n},p_K^n)$ $=
I(C^{(n)};$ $\bm{X})$ to derive the two 
bounds in Theorem~\ref{MI_FV}. Let 
$
h(q):=-q \log q -(1-q)\log (1-q),\: q \in [0,1],
$ 
be the binary entropy function.
The following lemma is useful for 
the computation of 
$
I(C^{(n)};$ $\bm{X})$.  
\begin{lemma}\label{lem:logOneMinuSq}
Let $q\in[0,1]$, $\eta \geq 1$, $\omega \geq 0$. 
We set $\Omega=\eta 2^{-\omega}$. 
If $q\leq \Omega$,
then we have 
$
h(q) \leq [\omega+\log e]\Omega.
$
\end{lemma}

Proof of Lemma \ref{lem:logOneMinuSq} 
is given in Appendix \ref{apd:logOneMinuSq}.
In the following we give the proof of Theorem~\ref{MI_FV}.

\textit{Proof of Theorem~\ref{MI_FV}:}
We evaluate the average codeword length and 
the mutual information for the coding scheme described above.
By Propositions~\ref{pos:Kobayashi} 
and~\ref{pos:upb_Max-MI}, 
$\forall n\in \mathbb{N}$, 
\begin{align}
& \Pr\{\bm{X} \in \mathcal{X}^n 
  - \mathcal{C}_{\empty}^{n}(R)\}
\leq (n+1)^{|\mathcal{X}|}
         2^{-nE_{\empty}(R|p_X)},
\label{eqn:ThOnePrfaa}\\
&m \log |\mathcal{X}| - H(\widetilde{K}^m) 
 \leq \left(R_n+\frac{1}{2}\right)(n+1)^{3|{\cal X}|}
      2^{-n[F(R|p_K)-\gamma_n]}
\notag\\
&\leq (R_n+1)|{\cal X}|(n+1)^{4|{\cal X}|}
2^{-nF(R|p_K)}. \label{eqn:ThOnePrfb}   
\end{align}
The above two bounds are quite useful for 
evaluating upper bounds of 
$\overline{L}_{\Phin}$ and $I(C^{(n)};\bm{X})$.   

\noindent \underline{\it Evaluation of 
the Average Codeword Length:} \/ We first 
derive an upper bound of the average codeword length.
Since $(\Phin, \Psin)$ satisfies (\ref{eqn:LengthEq}), 
we have
$\overline{L}_{\Phin}=\overline{L}_{\phin}^{\prime}$. 
On the upper bound of $\overline{L}_{\Phin}$,
we have the following chain of inequalities:
\begin{align*} 
 &
 \overline{L}_{\Phin}
 =\overline{L}_{\phin}^{\prime}
=\sum_{\bm{x} \in \mathcal{X}^n} p_{\bm{X}}(\bm{x}) L_{\phin}^{\prime}
\\
&
 = \sum_{\bm{x} \in \mathcal{C}_{\empty}^{n}(R)} 
    p_{\bm{X}}(\bm{x})\lceil nR_n \rceil
 + \sum_{\bm{x}\in 
    \mathcal{X}^n-\mathcal{C}_{\empty}^{n}(R)} 
    p_{\bm{X}}(\bm{x})
    \lceil n\log|\mathcal{X}| \rceil
   \\
 &\leq (nR_n+1)
 \Pr\{\bm{X}\in 
    \mathcal{C}_{\empty}^{n}(R)\}
 \\
 &\quad + (n\log |\mathcal{X}|+1) 
   \Pr\{\bm{X}\in \mathcal{X}^n
   -\mathcal{C}_{\empty}^{n}(R)\}
 \\
 &= nR_n+1+n(\log |\mathcal{X}| - R_n)
  \Pr\{\bm{X}\in \mathcal{X}^n
  -\mathcal{C}_{\empty}^{n}(R)\}
  \\
 &\MLeq{a} nR_n+1
   +n(\log|\mathcal{X}|- R)(n+1)^{|\mathcal{X}|} 
   2^{-nE_{\empty}(R|p_K)}.
\end{align*}
Step~(a) follows from $R_n\geq R$ and 
(\ref{eqn:ThOnePrfaa}).  

\noindent 
\underline{\it Evaluation of the Mutual Information:} \/ We next derive an upper bound of the mutual information.
For $I(C^{(n)};\bm{X})$, we have the 
following equalities:
\begin{align*} \label{eq:I_C;X0}
 &I(C^{(n)};\bm{X})
 = I(C^{(n)};\bm{X}, L^{\prime}_{\phin}(\bm{X})) \\
 &= I(C^{(n)}; L^{\prime}_{\phin}(\bm{X})
    )
    + I(C^{(n)};\bm{X} |L^{\prime}_{\phin}(\bm{X})
    ).
 \stepcounter{equation}\tag{\theequation}
\end{align*}
Set $q=\Pr\{\bm{X} \in \mathcal{X}^n - \mathcal{C}_{\empty}^{n}(R)\}$ and 
$\Omega=(n+1)^{|\mathcal{X}|}$ $ 
   2^{-nE_{\empty}(R|p_X)}$. Then the bound 
(\ref{eqn:ThOnePrfaa}) is equivalent to
\begin{align}
& q\leq \Omega=\eta 2^{-\omega},
  \eta= (n+1)^{|\mathcal{X}|}>1, 
  \omega=nE_{\empty}(R|p_X).
\label{eqn:qOmegaEtc}  
\end{align} 
For the first term in the right members of~(\ref{eq:I_C;X0}), 
we have
\begin{align*}\label{eq:FVconMI1}
 &I(C^{(n)}; L^{\prime}_{\phin}(\bm{X}))
 \leq 
 H(L^{\prime}_{\phin}(\bm{X}))
 =h(q)
 \\
 &\MLeq{a}\left\{
    nE_{\empty}(R|p_X)+\log e\right\}
 (n+1)^{|\mathcal{X}|}2^{-nE_{\empty}(R|p_X)}.
 \stepcounter{equation}\tag{\theequation}
\end{align*}
Step~(a) follows from (\ref{eqn:qOmegaEtc}) 
and Lemma~\ref{lem:logOneMinuSq}. 
For the second term in the right members of~(\ref{eq:I_C;X0}), we have the following equalities:
\begin{align*}\label{eq:FVconMI2}
 &I(C^{(n)};\bm{X}|L^{\prime}_{\phin}(\bm{X}))
 \\
 &= \Pr\{L^{\prime}_{\phin}=\lceil nR_n \rceil\}
    I(C^{(n)};\bm{X}|
    L^{\prime}_{\phin}(\bm{X})=\lceil nR_n \rceil )
  \\
 &\quad + \Pr\{ L^{\prime}_{\phin}(\bm{X})
    =\lceil n\log|\mathcal{X}|\rceil\}
 \\
 &\quad \times I(C^{(n)};\bm{X}|
 L^{\prime}_{\phin}
 (\bm{X})
 =\lceil n\log|\mathcal{X}| \rceil).
 \stepcounter{equation}\tag{\theequation}
\end{align*}
We first evaluate the first term in the 
right members of~(\ref{eq:FVconMI2}).
We have the following inequalities:
\begin{align*}\label{ineq:FVconMI1}
 &I(C^{(n)};\bm{X}|
 L^{\prime}_{\phin}(\bm{X})=\lceil nR_n \rceil )
 \\
&= H(\widetilde{K}^m \oplus \phinX|  
   L^{\prime}_{\phin}
   (\bm{X})
    =\lceil nR_n \rceil )\\
 &\quad-H(\widetilde{K}^m \oplus \phinX|\bm{X},
   L^{\prime}_{\phin}
   (\bm{X})
    =\lceil nR_n \rceil ) 
   \\
 &\leq m \log |\mathcal{X}|
  - H(\widetilde{K}^m|\bm{X},
    L^{\prime}_{\phin}
    (\bm{X})
    =\lceil nR_n \rceil)\\
 &\MEq{a}m\log |\mathcal{X}| 
         -H(\widetilde{K}^m).
 \stepcounter{equation}\tag{\theequation}
\end{align*}
Step (a) follows from $\bm{K}\perp \bm{X}$.
For the second term in the 
right members of~(\ref{eq:FVconMI2}), 
we have the following: 
\begin{align*}\label{ineq:FVconMI2}
 &I(C^{(n)};\bm{X}|L^{\prime}_{\phin}(\bm{X})
 =\lceil n\log|\mathcal{X}| \rceil)
 \leq n\log|\mathcal{X}|.
 \stepcounter{equation}\tag{\theequation}
\end{align*}
From~(\ref{eq:FVconMI2}), (\ref{ineq:FVconMI1}), and~(\ref{ineq:FVconMI2}), we have    
the following bound of the conditional mutual information:
\begin{align}
 &I(C^{(n)};\bm{X}|L^{\prime}_{\phin}(\bm{X}))
 \leq m \log |\mathcal{X}| - H(\widetilde{K}^m)
\notag\\ 
 &\quad +\Pr\{\bm{X} \in \mathcal{X}^n 
   - \mathcal{C}_{\empty}^{n}(R)\} n\log |\mathcal{X}|
\notag\\
 &\MLeq{a} (R_n+1)|{\cal X}|(n+1)^{4|{\cal X}|}2^{-nF(R|p_K)}
 \notag\\
 &\quad  + n\log |\mathcal{X}|(n+1)^{|\mathcal{X}|} 2^{-nE_{\empty}(R|p_X)}.
\label{ineq:FVconMI3}
 \end{align}
Step (a) follows from (\ref{eqn:ThOnePrfaa}) and (\ref{eqn:ThOnePrfb}). 
From~(\ref{eq:I_C;X0}), 
(\ref{eq:FVconMI1}), and~(\ref{ineq:FVconMI3}), 
we have 
$\forall n\in \mathbb{N}$ 
and $\forall (p_X,$ $p_K)\in {\cal P}^2({\cal X})$,
\begin{align*} 
 &I(C^{(n)};\bm{X}) \leq 
 \left\{nE_{\empty}(R|p_X)
 +\log e \right\}(n+1)^{|\mathcal{X}|}
  2^{-nE_{\empty}(R|p_X)}\\
 &\quad + (R_n+1)|{\cal X}|(n+1)^{4|{\cal X}|} 2^{-nF(R|p_K)}\\
 &\quad + n \log |\mathcal{X}|(n+1)^{|\mathcal{X}|}
 2^{-nE_{\empty}(R|p_X)}\\
 &= \left\{ n[E_{\empty}(R|p_X)   
    + \log |\mathcal{X}|]+\log e \right\}
    (n+1)^{|\mathcal{X}|} 2^{-nE_{\empty}(R|p_X)}\\
 & \quad + (R_n+1)|{\cal X}|(n+1)^{4|{\cal X}|}2^{-nF(R|p_K)},
\end{align*}
completing the proof. 
\hfill\IEEEQED

\subsection{Proof of Theorem 
\ref{th:converseTh_FV}}
\label{subsec:ThmConv}

In this subsection, we prove 
Theorem~\ref{th:converseTh_FV}.
To prove this theorem it suffices to show that under the assumption that $R$ is 
$\delta$-admissible we have 
$R \geq R^{\star}(p_X,p_K)$.
From (\ref{def:Rast2}), $R \geq R^{\star}(p_X,p_K)$ is equivalent to the following:
\begin{align}
  &R \geq H(X),\: 
  H(K) \geq H(X).\label{RKgeqX}
\end{align}

\newcommand{\OmitT}{
The proof of Theorem~\ref{th:converseTh_FV} assumes that $R$ is $\delta$-admissible. 
Under this assumption, proving that $R \geq R^{\ast}(p_X,p_K)$ is equivalent to proving Theorem~\ref{th:converseTh_FV}.
From (\ref{def:Rast2}), 
$R \geq R^{\ast}(p_X,p_K)$ 
is equivalent to the following:
\begin{align}
 &R \geq H(X), \label{RgeqX}\\
 &H(K) \geq H(X). \label{KgeqX}
\end{align}
}
\textit{Proof of Theorem~\ref{th:converseTh_FV}:}
We assume that 
$R$ is $\delta$-admissible. 
Then $\exists \{(\Phin,$ $\Psin)\}^{\infty}_{n=1}$ 
such that $\forall \gamma >0$, 
    $\exists n_0=n_0(\gamma) \in \mathbb{N}$, 
    $\forall n$ $\geq n_0$, 
\begin{align}
 &\frac{1}{n}
 \expectation\left[L_{\PhiKn}(\bm{X})\right] \leq R+\gamma, \:
 I(C^{(n)};\bm{X})\leq \delta. \label{gyaku2_fv}
\end{align}
We first prove $R\geq H(X)$ 
in (\ref{RKgeqX}). We have that $\forall n\geq n_0$, 
\begin{align}
 &R+\gamma\MGeq{a}
 \frac{1}{n}\expectation
 \left[L_{\PhiKn}(\bm{X})
 \right] 
\notag\\
 & 
 \MGeq{b} 
 H(X)-\frac{1}{n}\log n
     -\frac{1}{n}\log(2e\log|\mathcal{X}|).
\label{gyaku3_fv}
\end{align}
Step (a) follows from (\ref{gyaku2_fv}) and Step (b) follows from Proposition~\ref{pos:Ltilde}.
Taking $n \to \infty$ in~(\ref{gyaku3_fv}), 
we have $R +\gamma \geq H(X)$. 
Since $\gamma>0$ can arbitrary be small, 
we have $R \geq H(X)$.

We next prove $H(K)\geq H(X)$ in (\ref{RKgeqX}). 
We have that $\forall n\geq n_0$, 
\begin{align*}\label{lwb:MI}
 &\delta\MGeq{a} I(C^{(n)};\bm{X})
 = H(\bm{X})-H(\bm{X},\PhiXnK|C^{(n)})\\
 &\MGeq{b} H(\bm{X})-H({\bm X},\bm{K} |C^{(n)})
 \MEq{c}H(\bm{X})-H(\bm{K}|C^{(n)})\\
 &\geq H(\bm{X})-H(\bm{K})= n[H(X)-H(K)].
 \stepcounter{equation}\tag{\theequation}
\end{align*}
Step~(a) follows from (\ref{gyaku2_fv}).
Step~(b) follows from the data processing 
inequality.
Step~(c) follows from that by the construction of $(\Phin,$ $\Psin)$, $\bm{K}$ and $C^{(n)}$ 
uniquely determine $\bm{X}$. 
From~(\ref{lwb:MI}), we have
$H(K)\geq H(X)-(1/n)\delta$. 
Taking $n\to\infty$ in this bound.
we have $H(K)\geq H(X)$.\hfill\IEEEQED
%

\newcommand{\OmiRR}{
}{
\section{Conclusions}

We proposed a framework 
for the variable length source encryption.
Under this framework, we obtained the necessary and sufficient condition on $R$ and $(p_X,p_K)$ for simultaneously achieving communication efficiency and security. 
The obtained condition does not depend on the constant 
$\delta\in (0,\delta_0]$, implying that we have the 
strong converse coding theorem for the proposed
framework of source encryption.
We also establish a universal construction 
of $\left\{(\Phin,\Psin)\right\}_{n=1}^{\infty}$
attaining the condition in the sense that 
$\left\{(\Phin,\Psin)\right\}_{n=1}^{\infty}$   
does not depend on $(p_X,p_K) \in {\cal P}^2({\cal X})$.

Future works include extending the source encryption framework to the case of more general sources, and extending this framework to the 
variable-length lossy source coding problem 
where the distortion is allowed.
}

\bibliographystyle{IEEEtran}
\bibliography{isita2026_FF_FV}

\appendix
\renewcommand{\thesection}{Appendix~\Alph{section}} 

\newcommand{\OmiTApdA}{
\subsection{Proof of Lemma \ref{lem:scs_tilde}} 
\label{prf:scs_tilde}

In this appendix, we prove Lemma \ref{lem:scs_tilde}.
Lemma \ref{lem:scs_tilde} part a) follows 
directly from the definition of $\mathcal{D}_l$.
We prove the part b).
We denote the restriction of $\phin$ 
to $\mathcal{D}_l$, by $\phin_{\mathcal{D}_l}$. 
From the definition of $\mathcal{D}_l$, we have $|\phin(\mathcal{D}_l)| = |\mathcal{D}_l|$, 
from which, by setting 
$m_l = \lceil \log |\mathcal{D}_l| \rceil$, 
there exists a one-to-one mapping
\begin{align*}
  \tau^{(n)}:\phin(\mathcal{D}_l) \to \{0,1\}^{m_l}.
\end{align*}
The maps $\phin_{\mathcal{D}_l}$ and 
$\tau^{(n)}$ are shown in Fig. \ref{tau_shazou}.
\begin{figure}[t]
  \centering
  \includegraphics[width=84mm]{tildephi_shazou.eps}
  \caption{The mappings 
  $\widetilde{\phi}^{(n)}_{\mathcal{D}_l}$ and $\tau^{(n)}$} \label{tau_shazou}
\end{figure}
For each $\mathcal{D}_l$, we define 
the mapping $\widetilde{\phi}_{\mathcal{D}_l}^{(n)}: 
\mathcal{D}_l \to \{0,1\}^*$ by
\begin{align*}
  \widetilde{\phi}^{(n)}_{\mathcal{D}_l}
  = \tau^{(n)} \circ \phin_{\mathcal{D}_l}.
\end{align*}
From the construction of $\widetilde{\phi}^{(n)}_{\mathcal{D}_l}$, 
we have for $\bm{X} \in \mathcal{D}_l$,
\begin{align*}
  L^{\prime}_{\widetilde{\phi}^{(n)}}(\bm{X}) 
  = |\widetilde{\phi}^{(n)}(\bm{X})| 
  = m_l = \lceil \log |\mathcal{D}_l| \rceil,
\end{align*}
from which, we have for $\bm{X} \in \mathcal{D}_l$, 
\begin{align} \label{rng:tL'}
  \log |\mathcal{D}_l| \leq 
   L^{\prime}_{\widetilde{\phi}^{(n)}}(\bm{X})
  \leq \log |\mathcal{D}_l| +1.
\end{align}
Here, taking the expectation of 
$L^{\prime}_{\widetilde{\phi}^{(n)}}(\bm{X})$ 
with respect to  $\bm{X}$, we have
\begin{align*} \label{eq:E[L']}
  \overline{L}^{\prime}_{\widetilde{\phi}^{(n)}} 
  &= \sum_{\bm{x}} \Pr \{\bm{X} = \bm{x}\}
  L^{\prime}_{\widetilde{\phi}^{(n)}}(\bm{X})   
  \\                                    
 &=\sum_{l=1}^{l_{\textrm{max}}} \sum_{\bm{x} 
 \in \mathcal{D}_l} 
 \Pr \{\bm{X}=\bm{x}\} 
 L^{\prime}_{\widetilde{\phi}^{(n)}}(\bm{X}).
  \stepcounter{equation}\tag{\theequation}
\end{align*}
From (\ref{rng:tL'}) and (\ref{eq:E[L']}), for the upper and lower bounds of $\overline{L}^{\prime}_{\widetilde{\phi}^{(n)}}$, we have 
the following inequalities:
\begin{align}
  \overline{L}^{\prime}_{\widetilde{\phi}^{(n)}} 
  &\leq \sum_{l=1}^{l_{\textrm{max}}} \sum_{\bm{x} \in   
  \mathcal{D}_l} \Pr \{\bm{X} = \bm{x}\}
  (\log |\mathcal{D}_l| +1) 
 \notag \\
 &= \sum_{l=1}^{l_{\textrm{max}}} 
 p_{\bm{X}} (\mathcal{D}_l) 
 \log|\mathcal{D}_l|+1,
 \label{upb:tL'}\\
 \overline{L}^{\prime}_{\widetilde{\phi}^{(n)}} 
 &\geq \sum_{l=1}^{l_{\textrm{max}}} \sum_{\bm{x} 
 \in \mathcal{D}_l} \Pr \{\bm{X} = \bm{x}\} 
 \log |\mathcal{D}_l| 
 \notag\\
&= \sum_{l=1}^{l_{\textrm{max}}} 
   p_{\bm{X}} (\mathcal{D}_l)
   \log|\mathcal{D}_l|. \label{lwb:tL'}
\end{align}
Thus, from (\ref{upb:tL'}) and (\ref{lwb:tL'}), we have
\begin{align*}
  \sum_{l=1}^{l_{\textrm{max}}}
  p(\mathcal{D}_l) \log |\mathcal{D}_l| &\leq \overline{L}^{\prime}_{\widetilde{\phi}^{(n)}} \leq \sum_{l=1}^{l_{\textrm{max}}}
  p(\mathcal{D}_l) \log |\mathcal{D}_l| +1.
\end{align*}
We finally prove the first bound in the part b).   
From the part a), we have 
$\log |\mathcal{D}_l| \leq l=L^{\prime}_{\phin}(\bm{X})$,
which together with 
(\ref{upb:tL'}) yields the following: 
\begin{align*}
& \overline{L}^{\prime}_{\widetilde{\phi}^{(n)}} \leq  
  \sum_{l=1}^{l_{\textrm{max}}} 
   p_{\bm{X}}(\mathcal{D}_l)L^{\prime}_{\phin}(\bm{X}) +1 \\
&= \sum_{l=1}^{l_{\textrm{max}}} 
   \sum_{\bm{x} \in \mathcal{D}_l} 
   \Pr \{\bm{X}= \bm{x}\} L^{\prime}_{\phin}(\bm{X})+1 
 = \overline{L}^{\prime}_{\phin} + 1. 
\end{align*}
Hence, we have the part b).
\hfill\IEEEQED
}

\subsection{Proof of Proposition 
\ref{pos:Ltilde}}\label{prf:pos_Ltilde}
In this appendix, we give the proof 
of Proposition~\ref{pos:Ltilde}.
We first present a lemma necessary for the proof.

\begin{lemma}\label{lem:HLupb}
Let $L$ be a positive integer-valued random variable with mean
$\mu=\expectation[L]$.
Then,
\begin{align}
 H(L)& \leq 
 \mu\log\mu-(\mu-1)\log(\mu-1)
\notag\\
& =\mu h(\mu^{-1})
     < \log(e\mu).
\label{eqn:HLUpB}
\end{align}
\end{lemma}

{\it Proof:} The second equality in (\ref{eqn:HLUpB}) is obvious. The third inequality 
follows from 
$$
\mu h(\mu^{-1})
=\log \mu +(\mu-1)\log \left(1+\frac{1}{\mu-1}\right)
<\log \mu+ \log e.
$$
Hence it suffices to prove the first inequality.  
Let $p_L=\{p_{L}(l)\}_{l\geq 1}$. 
By definition it is obvious 
that $\mu\geq 1$. 
When $\mu=1$, we have $H(L)=0$. Hence we have the bound of Lemma \ref{lem:HLupb}. When $\mu>1$, let 
$\wt{L}$ be another positive integer-valued 
random variable having the geometrical distribution 
given by 
$p_{\wt{L}}(l)=(\mu-1)^{-1}
\left(1-\mu^{-1}\right)^{l}, l\geq 1. $  
Then we have the following:
\begin{align*}
0&\leq D(p_L||p_{\wt{L}})
=-H(L)-\sum_{l\geq 1}p_{L}(l)\log p_{\wt{L}}(l) 
\\
&=-H(L)+\mu\log\mu-(\mu-1)\log(\mu-1),
\end{align*}
completing the proof. 
\hfill\IEEEQED

{\it Proof of Proposition \ref{pos:Ltilde}:}
Fix $\bm{k}\in {\cal X}^n$ arbitrary. Set  $\underline{Y}_{\bm{k}} \coloneqq 
 \PhiknX$, which is a binary random 
sequence with the form:  
$
 \underline{Y}_{\bm{k}}
 = Y_{\bm{k},1}Y_{\bm{k},2}
 \cdots Y_{\bm{k},\LPhiknX}.
$
On upper bound of $H({\bm X})$, we have the 
following chain of inequalities: 

\begin{align*} \label{ineq:H(X)}
&nH(X)=H(\bm{X}) \MEq{a}H\left(\underline{Y}_{\bm{k}}\right) 
 =H\left(\underline{Y}_{\bm{k}}, \LPhiknX \right) 
\\
&= H\left(\LPhiknX \right)+ 
\sum_{l\geq 1}
    \Pr\{\LPhiknX =l\}
\\
&\quad \times 
H\left(
{Y}^l_{\bm{k}} \left|
  \LPhiknX=l\right.\right)
 \\
 & \MLeq{b} 
H\left( \LPhiknX \right) 
 + \sum_{l\geq 1}
  l\Pr\{ \LPhiknX=l\}
\\ 
&= H\left( \LPhiknX \right) 
  + \overline{L}_{\Phi^{(n)}_{\bm{k}}}.
  \stepcounter{equation}\tag{\theequation}
\end{align*}
Step (a) follows from 
$\Phi_{\bm{k}}^{(n)}$ 
is bijective. 
Step (b) follows from that since 
${Y}_{\bm{k}}^l \in \{0,1\}^l$, 
we have $H({Y}_{\bm{k}}^l \mid 
L_{{\PhiknX}}=l) \leq l$. 
From (\ref{ineq:H(X)}), we have 
\begin{align} \label{ineq:ELtent}
  \overline{L}_{\Phi^{(n)}_{\bm{k}}}
  \geq nH({X}) -H\left(\LPhiknX \right).
  \end{align}
We next evaluate the second term in the right members 
of (\ref{ineq:ELtent}). We first assume that 
$ \overline{L}_{\Phi^{(n)}_{\bm{k}}}
=\expectation\left[\LPhiknX
\right] \geq n\log|\mathcal{X}|.
$
In this case we have 
$
\overline{L}_{\Phi^{(n)}_{\bm{k}}} 
\geq n\log|\mathcal{X}|\geq nH(X).
$
Hence we have the bound (\ref{eqn:LtildeLbTwo}) 
in Proposition \ref{pos:Ltilde}.
We next assume that 
$\overline{L}_{\Phi^{(n)}_{\bm{k}}}
  < n\log|\mathcal{X}|$.
Then, by Lemma~\ref{lem:HLupb},
\begin{align*} 
& H\left({{\LPhiknX}}\right)
< \log\left(
  e\overline{L}_{\Phi^{(n)}_{\bm{k}} }\right)
  <\log (n e\log|\mathcal{X}|), 
\end{align*}
which together with~(\ref{ineq:ELtent}) 
yields that
\begin{align}
\overline{L}_{\Phi^{(n)}_{\bm{k}} }
 \geq nH({X})-\log n-\log(e\log|\mathcal{X}|).
\label{eqn:upbHL}
\end{align}
Taking expectation of both sides of (\ref{eqn:upbHL})
with respect to $p_{\bm{K}}$, we obtain the bound 
(\ref{eqn:LtildeLbTwo}) 
in Proposition \ref{pos:Ltilde}.
\hfill\IEEEQED

\subsection{Proof of Lemma \ref{lem:logOneMinuSq}}
\label{apd:logOneMinuSq}
For each $q\in [0,1)$, we have the following: 
\begin{align*}
  &(\log e)^{-1}(1-q)
  \log \frac{1}{1-q}
  =(1-q)\ln \frac{1}{1-q}
  \\
  & =(1-q)\sum_{k=1}^{\infty}\frac{q^k}{k}
  \leq (1-q)
  \left(q +\sum_{k=2}^{\infty}\frac{q^k}{2}
  \right)
  =q-\frac{1}{2}q^2\leq q, 
\end{align*}
from which we have that for each $q \in [0,1)$, 
\begin{align}
h(q)\leq [-\log q +\log e]q. 
\label{eqn:barQlogbarQ}
\end{align}
For $q \geq 0$, we set 
$g(q) \coloneqq [-\log q+\log e]q.$ Since 
$$
\frac{\rm d}{{\rm d}q}g(q)=-\log q,
$$
$g(q)$ is monotone increasing for $q \in [0,1]$. 
We first consider the case where 
$0\leq q \leq 1\leq \Omega$. In this case we have
$$
h(q)\leq 1 \leq \Omega \leq [\omega+ \log e]\Omega.
$$
We next consider the case where 
$0\leq q \leq\Omega < 1$. 
In this case we have
\begin{align*}
h(q)&\MLeq{a}[-\log q+\log e]q
\MLeq{b}[-\log \Omega+\log e]\Omega
\\
&=[\omega-\log \eta +\log e]\Omega
\MLeq{c}[\omega+\log e]\Omega.
\end{align*}
Step (a) follows from (\ref{eqn:barQlogbarQ}).
Step (b) follows from that $g(a)$ 
is a monotone increasing function of 
$a \in[0,1)$. Step (c) follows from 
$\eta \geq 1$. \hfill\IEEEQED

\end{document}